\documentclass{jfm}
\usepackage{indentfirst}
\usepackage{amsmath}
\usepackage{amsfonts}
\usepackage{amssymb}
\usepackage{graphicx}
\usepackage{titlesec}
\usepackage{bbold}
\usepackage{tikz}
\usepackage{hyperref}

\newtheorem{theorem}{Theorem}
\titleformat{\section}
  {\normalfont\fontsize{14}{15}\bfseries}{\thesection}{1em}{}

\title{Approximating pressure-driven Stokes flow using the principle of minimal excess dissipation}

\author{Tachin Ruangkriengsin\aff{1}
  \corresp{\email{bankcrub@g.ucla.edu}}
 \and Marcus Roper\aff{1,2}}

\affiliation{\aff{1}Department of Mathematics, University of California Los Angeles, Los Angeles, California 90095, USA
\aff{2}Dept. of Computational Medicine, University of California Los Angeles, Los Angeles, California 90095, USA}
\begin{document}

\maketitle


\begin{abstract}
Stokes' equations model microscale fluid flows including the flows of nanoliter-sized fluid samples in lab-on-a-chip systems. Helmholtz's dissipation theorem guarantees that the solution of Stokes' equations in a given domain minimizes viscous dissipation among all incompressible vector fields that are compatible with the velocities imposed at the domain boundaries. Helmholtz's dissipation theorem directly guarantees the uniqueness of solutions of Stokes flow, and provides a practical method for calculating approximate solutions for flow around moving bodies. However, although generalization of the principle to domains with mixtures of velocity and stress boundary conditions is relatively straightforward (Keller et al., 1967), it appears to be little known. Here we show that the principle of minimal excess dissipation can be used to derive accurate analytical approximations for the flows in microchannels with different cross-section shapes including when the channel walls are engineered to have different distributions of slip boundary conditions. In addition to providing a simple, rapid, method for approximating the conductances of micro-channels, analysis of excess dissipation allows for a comparison principle that can be used, for example, to show that the conductance of a channel is always increased by adding additional slip boundary conditions.
\end{abstract}
    
    
\maketitle
    
\section{Introduction}\label{Intro}
In microchannel flows, small length scales and moderate flow speeds mean that viscous stresses dominate over inertial forces \citep{squires2005microfluidics,tabeling2005introduction}. In this setting, the flows, $\mathbf{u}$ and pressure $p$, can be modeled using the Stokes' equations; $\mu \nabla^2 \mathbf{u} +\mathbf{F} - \nabla p =0$ and $\nabla \cdot \mathbf{u}=0$, where $\mathbf{F}$ is the body force, and $\mu$ is the viscosity of the fluid. Although the linearity and lack of time dependence in these equations allow for drastically simpler methods of solution than the fully nonlinear Navier-Stokes equations, there are relatively few available exact solutions of the equations. Here, we consider pressure-driven flows in long channels with constant cross-sections. In addition to being common in lab on a chip device, such microchannel flows appear in models of air flow in lung alveoli \citep{bates2009lung} and insect feeding \citep{Insect_feeding}. Analytic solutions are possible when the channel walls align with the level curves of a coordinate system in which Laplace's equation is separable \citep{moon2012field,bahrami2005pressure,bazant2016exact,bruus2008theoretical}, as well as in a few other geometries such as equilateral triangular channels \citep{triangular}. However, even when analytical expressions for the flow can be derived, they often take the form of infinite series expansions, providing only indirect access into how channel conductance depends upon the parameters that control the channel shape. \citet{bahrami2005pressure} present closed form approximate solutions for the flow-pressure drop relationship in channels of constant cross-section, based on physical reasoning into the roles played by channel perimeter and moment of inertia, and show empirically that these approximate expressions closely model the exact conductance laws for rectangular, trapezoidal, semicircular and sector shaped channels. The agreement presented there is very good (within 10\% for all the channel shapes studied). However, despite its empirical success, no rigorous error bound is proven for the method, nor is there a path to make approximations of arbitrarily high accuracy or to generalizing the method to calculate conductance laws when the flow in the channel is not translationally invariant.

Advances in surface engineering now allow for the creation of channels in which the walls of the channel contain a mixture of no-slip and slip boundary conditions \citep{Lauga2007,neto2005boundary,tretheway2004generating}. In particular, nearly perfect slip can be produced on a boundary using a combination of surface chemistry and structure to trap a layer of gas bubbles \citep{NanobubbleIshida,tyrrell2001images,Tyrrell2002AtomicFM}. Inclusion of patterns of slip and no-slip on channel walls greatly complicates the mathematical challenge of solving Stokes equations. Previous analyses have exposed conductance relations for circular channels with parallel bands of slip that are oriented either longitudinally to transversely to the flow \citep{lauga2003effective} and for plane boundaries over which slip islands are randomly distributed \citep{sbragaglia_prosperetti_2007}. Approximate conductance laws are obtained from such models as asymptotic limits of exact solutions, and give insights into the quantitative dependence of conductance upon the sizes of slip patches and upon their distribution. 

Since obtaining exact solutions is frequently challenging, we are motivated to obtain a general framework for approximating pressure-driven Stokes flow that can be applied to microchannels with different cross-sections and boundary conditions. Our starting place is the well-known \textit{Helmholtz minimum dissipation theorem} \citep{helmholtz,happel2012low}, which states that the steady Stokes flow of an incompressible fluid has the smallest rate of dissipation among all incompressible velocity fields with the same velocity on the domain boundaries. The proof of this theorem is contained in our derivation of the principle of minimum excess dissipation in Section \ref{Method}. In his original statement of the theorem, \citet{helmholtz} showed that the solution of Stokes' equations are critical points of the dissipation and asserted that this critical point was the global minimum; \citet{korteweg1883xvii} provided an early proof that the Stokes flow solution is the global minimum. 

Helmholtz's theorem has two important application areas: first, it allows for a direct proof that the solution of Stokes' equations, with velocity boundary conditions specified, is unique for bounded domains. For if $\mathbf{u}_1$ and $\mathbf{u}_2$ are both solutions of Stokes equations, with the same velocity boundary conditions, then $\mathbf{v}\equiv \mathbf{u}_1-\mathbf{u}_2$ is also a solution of Stokes flow equations, that vanishes on all of $\Omega$'s boundaries. The unique minimizer of the dissipation is then $\mathbf{v} = \mathbf{0}$, assuring that $\mathbf{u}_1=\mathbf{u}_2$.

Second, the minimal dissipation theorem provides a comparison principle that allows for approximate computation of flow properties. For example, given a body $V$, traveling through a fluid at speed $U$, the force $F$ that needs to be applied in the direction of the body motion, is bounded below by the force needed to propel, at the same speed $U$, any body $V_1$ that can be inscribed in $V$, and above by any body, $V_2$, that contains $V$. We can obtain an incompressible flow field that satisfied the velocity boundary conditions on the surface of $V$, by taking the Stokes flow, $\mathbf{u}_2$ surrounding the outscribed body, $V_2$, and then continuously extending down to the boundary of $V$ by assigning the fluid in $V_2-V$ the rigid body speed of $V$. Since no dissipation occurs in the rigid body flow in $V_2-V$, the dissipation of this newly constructed flow is the same as for the flow around $V_2$. Hence, the dissipation of the flow around $V_2$ must exceed the dissipation in the flow field around $V$. Since the dissipation due to a body is the product of its speed and the component of the force propelling it in the direction of its velocity, and the velocities of $V$ and $V_2$ are equal, it follows that the component of the force propelling $V_2$ in the direction of its motion, must exceed the component of the force propelling $V$ in its direction of motion. By choosing $V_1$ and $V_2$ both spheres or other simple body shapes, this argument readily gives analytical upper and lower bounds upon the drag coefficient of any body in Stokes' flow \citep{hill1956extremum}. The same argument ensures that the drag force upon any moving body is increased in the presence of external bodies or boundaries. Minimum dissipation has also been used to bound the energetic efficiency of microswimmers propelled by surface tractions \citep{nasouri2021minimum}. 

Although the original theorem of Helmholtz included only the prescribed velocity on the boundary, modification to include boundaries on which stresses, rather than velocities, are specified is quite straightforward though not widely known. \citet{keller1967extremum} are, to our knowledge, the first to publish a general version of the theorem; in which they showed that Stokes flows in domains whose boundaries include velocity and traction boundary conditions, as well as boundaries with a mixture of the two minimize the excess dissipation; a quantity formed by subtracting off from the dissipation twice the rate of working done by the surface tractions and body forces. They used their theorem to prove results about the rheology of fluids containing solid bodies or droplets. More recently, one of us (Roper) used the discrete form of minimum excess dissipation (which, unaware of \citet{keller1967extremum}, we called the complementary dissipation function, based on its analogy with the complementary energy of externally loaded elastic body), to develop theory for the optimization of fluid transport networks upon which general boundary conditions of flow and pressure were simultaneously applied \citep{chang2018}.

Building upon these works, we investigate the use of minimum excess dissipation to approximate the pressure-driven Stokes flow in microchannels using more complex boundary conditions. The remainder of this paper is organized as follows. In Section \ref{Method}, we revisit the formulation of excess dissipation and our approximation frameworks for the pressure-driven Stokes flow. In Section \ref{Triangular_Rectangular}, we approximate the pressure-driven Stokes flow in channels with triangular and rectangular cross-sections with no-slip boundary conditions prescribed on all channel walls. In Section \ref{Lagrange_Section}, we introduce a modification to the excess dissipation with an additional penalty term to expand the space of valid test functions. In Section \ref{Mixed_Boundary}, we extend our result to larger families of boundary conditions; in particular, we consider channels in which walls are patterned into bands of no-slip and of slip surfaces, and produce approximate analytical conductance formulas that closely replicate the exact results presented in \citet{lauga2003effective}.

\section{Deriving the Principle of Excess Dissipation}\label{Method}

In this paper, we analyze the Stokes flow of an incompressible fluid, which approximates the fluid motion in bounded domains in the limit of low Reynolds number. For our purposes, it is convenient to write these equations in their \textit{weak form}: Suppose $\mathbf{u}$ satisfies the Stokes' equations in a domain $\Omega$, boundary $\partial \Omega$, and the corresponding strain-rate tensor is $\mathbf{E} = \frac{1}{2}(\nabla \mathbf{u} + (\nabla \mathbf{u})^{T})$. Then, if $\mathbf{u}^{\star}$ is any continuously differentiable and incompressible vector field on the same domain, and $\mathbf{E}^{\star} = \frac{1}{2}(\nabla \mathbf{u}^{\star} + (\nabla \mathbf{u}^{\star})^{T})$ is the corresponding strain-rate tensor, then 
 \begin{equation}\label{Weak_form_stokes}
    \int_{\partial \Omega} \mathbf{n}\cdot \boldsymbol{\sigma} \cdot \mathbf{u}^{\star} \: dS + \int_{\Omega} \mathbf{F}\cdot\mathbf{u}^{\star} \: dV = 2\mu \int_{\Omega} \mathbf{E}:\mathbf{E}^{\star} \: dV
 \end{equation}
Here, $\boldsymbol{\sigma}$ is the stress tensor and the right-hand-side of equation (\ref{Weak_form_stokes}) may be written in index notation as $2\mu \int_{\Omega}E_{ij} E^{\star}_{ji}\:dV$. Henceforth, we will assume that no body forces act within $\Omega$; $\mathbf{F}=\mathbf{0}$. Then, we follow \citet{keller1967extremum,chang2018}, by defining: the \textbf{excess dissipation} of a flow in the control volume $\Omega$ whose boundary $\partial \Omega$ may be partitioned into sets $\partial \Omega_T$ on which stress boundary conditions $\bf{n} \cdot \boldsymbol{\sigma}  = T$ are prescribed, and $\partial \Omega_U$ on which velocity boundary conditions $\bf{u} = U$ are prescribed is defined by
    \begin{equation}\label{Comp_Dis}
        \mathcal{H}[\mathbf{u}] := 2\mu \int_{\Omega} \mathbf{E}:\mathbf{E} \: dV - 2\int_{\partial\Omega_T} \mathbf{T}\cdot \mathbf{u} \: dS
      \end{equation}
The first term in equation (\ref{Comp_Dis}) represents the total rate of viscous dissipation by fluid flow in $\Omega$, and the surface integral term represents twice the rate of working of external tractions applied at $\partial \Omega_T$. If the excess dissipation is evaluated for a solution to Stokes' equations, then the rate of dissipation is equal to the rate of work done by the boundary forces. That is $2\mu \int_{\Omega} \mathbf{E}:\mathbf{E} \: dV = \int_{\partial \Omega} \mathbf{n}\cdot \mathbf{\sigma} \cdot \mathbf{u} \: dS$, so for, for these solutions, we obtain $H[\mathbf{u}] = \int_{\partial \Omega_U} \mathbf{n}\cdot \mathbf{\sigma} \cdot \mathbf{U} \: dS - \int_{\partial \Omega_T} \mathbf{n}\cdot \mathbf{T} \cdot \mathbf{u} \: dS$, i.e. the excess dissipation is equal to the work done by boundaries on which velocities are prescribed, minus the work done by the applied surface tractions. Our central approximation framework originates in the minimization of excess dissipation:
\begin{theorem}\label{Excess_dissipation_thm}
(Keller et al., 1967) The Stokes flow of an incompressible fluid minimizes the excess among all incompressible vector fields that satisfy the velocity boundary conditions on $\partial \Omega_U$. 
\end{theorem}
\begin{proof}
    Let $\bf{u}$ be the Stokes flow and $\bf{u^{\star}}$ be any incompressible vector field with $\bf{u^{\star}} = U$ on $\partial \Omega_U,$ and define $\mathbf{E}^{\star} = \frac{1}{2}(\nabla \mathbf{u}^{\star} + (\nabla \mathbf{u}^{\star})^{T})$ to be the corresponding strain-rate tensor. As $\mathbf{E}$ and $\mathbf{E}^{\star}$ are both symmetric, we have:
    \begin{equation}
        2\mu \int_{\Omega} (\mathbf{E}-\mathbf{E}^{\star}):(\mathbf{E}-\mathbf{E}^{\star}) \: dV \geq 0~~, \label{eq:1stinequality}
    \end{equation}
    or, equivalently, on expanding the integrand
    \begin{equation}\label{Stoke_Ie}
        2\mu \int_{\Omega}\mathbf{E}^{*}:\mathbf{E}^{*} \: dV - 2\mu \int_{\Omega} \mathbf{E}:\mathbf{E} \: dV \geq 4\mu \int_{\Omega} \mathbf{E}:(\mathbf{E}^{\star}-\mathbf{E}) \: dV
    \end{equation}
    Also, using the test velocity $\bf{u^{*}}-\bf{u}$ in the weak form of Stokes' equation (\ref{Weak_form_stokes}) gives 
    \begin{equation}\label{Stoke_Weak}
        \int_{\partial \Omega} \textbf{n} \cdot \boldsymbol{\sigma} \cdot (\mathbf{u^{\star} - u}) \: dS + \int_{\Omega} \mathbf{F\cdot(u^{\star}-u)} \: dV = 2\mu \int_{\Omega} \mathbf{E}:(\mathbf{E}^{\star}-\mathbf{E}) \: dV
    \end{equation}
    Putting equations \ref{Comp_Dis}, \ref{Stoke_Ie} and \ref{Stoke_Weak} together with zero body forces ($\bf{F} = 0$), we conclude:
    \begin{align}
        \mathcal{H}[\mathbf{u^{*}}]-\mathcal{H}[\mathbf{u}] &= 2\mu \left(\int_{\Omega} \mathbf{E}^{*}:\mathbf{E}^{*} \: dV- \int_{\Omega} \mathbf{E}:\mathbf{E} \: dV \right) - 2\left(\int_{\partial\Omega_T} \mathbf{T}\cdot \mathbf{u^{*}} \: dS- \int_{\partial\Omega_T} \mathbf{T}\cdot \mathbf{u} \: dS\right) \\
        &\geq 4\mu \int_{\Omega} \mathbf{E}:(\mathbf{E}^{\star}-\mathbf{E}) \: dV - 2 \int_{\partial\Omega_T} \mathbf{T}\cdot (\mathbf{u^{\star}-u}) \: dS \notag \\
        &= 2\int_{\partial \Omega} \textbf{n} \cdot \sigma \cdot (\mathbf{u^{\star} - u}) \: dS - 2 \int_{\partial\Omega_T} \mathbf{T}\cdot (\mathbf{u^{\star}-u}) \: dS \notag \\ 
        &= 2\int_{\partial \Omega_U} \textbf{n} \cdot \sigma \cdot (\mathbf{U-U}) \: dS +  2\int_{\partial \Omega_T} \mathbf{T} \cdot (\mathbf{u^{\star} - u}) \: dS- 2 \int_{\partial \Omega_T} \mathbf{T}\cdot (\mathbf{u^{\star}-u}) \: dS \notag \\
        &= 0 \notag
    \end{align} 
    Hence, $\mathcal{H}[\mathbf{u^{*}}] \geq \mathcal{H}[\mathbf{u}].$ Therefore, the Stokes flow minimizes the excess dissipation. 
\end{proof}
Note that equality is achieved in Eqn. \ref{eq:1stinequality}, when $\mathbf{E}^\star \equiv \mathbf{E}$; that is, only when the strain rates of the flows are identical; thus, the minimizers of $\mathcal{H}$ are Stokes flow plus any rigid body motion; under typical velocity boundary conditions that limit the rigid body motion component, the minimizer  is unique.

In the proceeding sections, we approximate the Stokes flow in channels with different cross-sections and boundary conditions by minimizing excess dissipation over a family of test functions. Formally, for each configuration, we posit a family of smooth test functions $\mathbf{u}_{\text{test}} = \mathcal{F}[c_1,c_2,...,c_n]$ where $\{c_i\}$ are parameters. The values of $\{c_i\}$ are chosen to minimize the functional $\mathcal{H}[\mathbf{u}_{\text{test}}]$; this can be done by setting equal to $0$ each of the partial derivatives of $\mathcal{F}$ with respect to the variables $\{c_i\}$ and solve for the corresponding system of equations. For the original excess dissipation formulation, we require that the test functions $\mathcal{F}$ satisfy the prescribed velocity conditions on the boundary. In Section \ref{Lagrange_Section}, we relax the velocity conditions on the boundary by introducing additional Lagrange multiplier terms to the excess dissipation. Throughout, we approximate the conductivity of the channel; that is the total flow across any cross-section of channel, given an applied unit pressure gradient. We compare our approximate channel conductances with existing analytical solutions and with conductances obtained by solving the Stokes equations using the commercial Finite Element package; COMSOL Multiphysics (Los Angeles, CA).

\section{Approximating flow in a channel with different cross-sections and no-slip boundary conditions}\label{Triangular_Rectangular}
\subsection{Theory for a prismatic channel}\label{Prismatic}
We consider pressure-driven Stokes flow of incompressible fluid in a channel of length $L$ with a constant cross-section shape. Here, we prescribe the pressure with $\mathbf{n} \cdot \boldsymbol{\sigma} = -\Delta p \mathbf{n}$ at $z = 0$ and $\mathbf{n} \cdot \boldsymbol{\sigma} = 0$ at $z = L$. Because of the symmetries in the problem setting, the velocity field is translational invariant along the longitudinal axis. Incompressibility is automatically satisfied since $\nabla \cdot \mathbf{u} = \frac{\partial u}{\partial z} = 0.$ To satisfy the balance of forces, we require: $p = p(z)$ for the $x-$ and $y-$components and 
\begin{equation}\label{Prismatic_Deriv}
    \mu \nabla^2 u = \frac{\partial p}{\partial z} = -\frac{\Delta p}{\mu L}
\end{equation}
where $\Delta p/L$ is the applied pressure gradient. Thus, the excess dissipation minimized by $u$ becomes:
\begin{equation}\label{Excess_Prismatic}
    \mathcal{H}[\mathbf{u}]=\mu L \int_{A}\left(\left(\frac{\partial u}{\partial x}\right)^2+ \left(\frac{\partial u}{\partial y}\right)^2\right) \: dxdy - 2\Delta p \int_{A}u \: dxdy
\end{equation}
Our goal is to calculate the total flux $Q = \int_{A} u \: dxdy$ down the channel as a function of $\frac{\Delta p}{L}$. In general, $Q= C\frac{\Delta p}{L}$ with constant of proportionality $C$, called the \textbf{channel conductivity}, that depends on channel shape.
\subsection{Rectangular cross-section channel}\label{Rectangular}
In this subsection, we consider the pressure-driven Stokes flow of incompressible fluid in a channel of length $L$ with a rectangular cross-section of width $l$ and height $h$. The actual Stokes flow satisfies the Poisson equation inside the rectangular region $\{(x,y): |x| \leq \frac{l}{2}, \quad |y| \leq \frac{h}{2}\}$ in the $x-y$ plane. To construct a family of test functions that vanish along the boundary of the rectangle, we posit $u(x,y) = c\left(\frac{l^2}{4}-x^2\right)\left(\frac{h^2}{4}-y^2\right)$ where $c$ is a constant parameter to be determined by minimizing the excess dissipation. For this assumed velocity profile, we may calculate
\begin{align}\label{Rectangular_Dissipation}
    \mathcal{H}[\mathbf{u}] = 
    \frac{h^3l^3(h^2+l^2)\mu L}{90}c^2 - \frac{h^3l^3\Delta p}{18}c
\end{align}
$\mathcal{H}$ is therefore quadratic in $c$, and has a unique minimizing value of $c$ when $c = \frac{5}{2(h^2+l^2)}\frac{\Delta p}{\mu L}$. In addition, we could calculate the channel conductance at this value of $c$ as:
\begin{equation}
Q = \frac{5h^3l^3}{72(h^2+l^2)}\frac{\Delta p}{\mu L}~. \label{eq:rectangular_conduct}
\end{equation} We may compare this fluid conductance with exact solution of the Poisson equation, obtained using Fourier series \citep{triangular} (Fig. \ref{Rectangular_Cond}):
\begin{equation}\label{exact_rectangular_conductance}
    Q_{\text{exact}} = \left(\frac{h^3l}{12}-\frac{16h^4}{\pi^5}\sum_{n=0}^{\infty}\frac{\cosh{\frac{(2n-1)\pi l}{h}}-1}{(2n-1)^5\sinh{\frac{(2n-1)\pi l}{h}}}\right)\frac{\Delta p}{\mu L}
\end{equation}
The approximation is very good when the height and width are close to each other. In a particular case when the channel is a square $l = h$, our approximation gives $Q = 0.0346\frac{\Delta p h^4}{\mu L}$ which is within $1.5\%$ of the exact channel conductance $Q_{\text{exact}} = 0.0351\frac{\Delta p h^4}{\mu L}$. In real microfluidic channels aspect ratios, $l/h$, may reach $4-8$, over which the error in our conductance estimate increases to $12.28\%$ (Fig. \ref{Rectangular_Cond}). At these high aspect ratios, the error in conductance is probably due to assuming a velocity profile that is quadratic in $x$ as well as in $y$. For sufficiently wide, slot-like, channels, we would expect the velocity profile to be approximately uniform in $x$, except within $\mathcal{O}(h)$ distance of the channel side walls. Although choosing test functions with a more complex $x-$dependence would allow a closer approximation to the real velocity field, they would lose the benefit of having simple-to-evaluate polynomial integrals in the excess dissipation. In Section \ref{Lagrange_Section}, we extend the minimization function $\mathcal{H}$ into a principle of optimization over functions that do not necessarily satisfy all boundary conditions.
\begin{figure}
\centering
\includegraphics[width = 0.5 \textwidth]{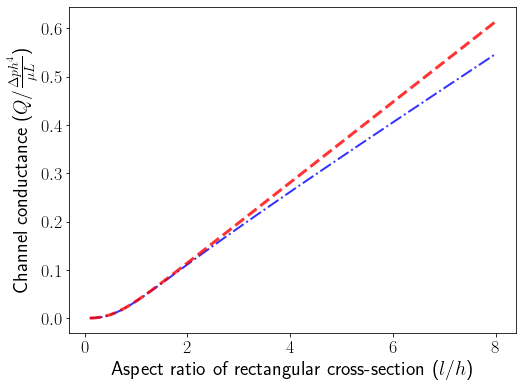}
\caption{In a rectangular channel, a quadratic velocity profile produces conductance estimates (blue dot-dashed curve) that agree closely with the exact conductance (red dashed curve, Eqn. \ref{exact_rectangular_conductance}) for channel aspect ratios $\lesssim 4$.}
\label{Rectangular_Cond}
\end{figure} 
\subsection{Triangular cross-section channel}\label{Triangular}
Similar to Section. \ref{Rectangular}, we consider the pressure-driven Stokes flow of incompressible fluid in a channel of length $L$ with triangular cross-section. By translating, rotating, and reflecting as necessary, we may assume that the vertices of the triangles are $(x_1 = 0,y_1 = 0), (x_2,y_2 = 0),$ and $(x_3,y_3)$ with $x_2 > 0, y_3 > 0.$ To construct test functions which vanish along the edges of the triangle, we move our analysis to barycentric coordinates: 
\begin{align}
    \lambda_1 &= \frac{(y_2-y_3)(x-x_3)+(x_3-x_2)(y-y_3)}{(y_2-y_3)(x_1-x_3)+(x_3-x_2)(y_1-y_3)} \notag \\ 
    \lambda_2 &= \frac{(y_3-y_1)(x-x_3)+(x_1-x_3)(y-y_3)}{(y_2-y_3)(x_1-x_3)+(x_3-x_2)(y_1-y_3)} \\ 
    \lambda_3 &= 1-\lambda_1-\lambda_2 \notag
\end{align}
We then consider a family of test functions of the form $u = c\lambda_1\lambda_2\lambda_3$ where the constant $c$ must be determined. Since we have $\lambda_i = 0$ for some $i \in \{1,2,3\}$ along each edge of the triangle, our test function obeys no-slip boundary conditions wherever it is required to. Here, we may calculate:
\begin{align}\label{Triangular_Dissipation}
    \mathcal{H}[\mathbf{u}] &= 2\mathcal{A} \mu L \int_{0}^{1}\int_{0}^ {1-\lambda_2} \left(\left(\frac{\partial u}{\partial x}\right)^2+\left(\frac{\partial u}{\partial y}\right)^2\right) \: d\lambda_1d\lambda_2- 4\mathcal{A}\Delta p \int_{0}^{1}\int_{0}^ {1-\lambda_2} u \: d\lambda_1d\lambda_2 \notag \\
    &= \left(\frac{(x_2^2-x_2x_3+x_3^2+y_3^2)\mu L}{180x_2y_3}\right)c^2 - \frac{x_2y_3 \Delta p}{60}c = \frac{(d_1^2+d_2^2+d_3^2)\mu L}{720 \mathcal{A}}c^2 - \frac{\mathcal{A}\Delta p}{30}c
\end{align}
where $\mathcal{A}$ is the area of the triangle and $d_1,d_2,d_3$ are the side lengths of the triangle. Thus, $\mathcal{H}$ admits a minimum when $c = \frac{12\mathcal{A}^2}{(d_1^2+d_2^2+d_3^2)}\frac{\Delta p}{\mu L}$ and a flow-pressure relation of $Q = \frac{\mathcal{A}^3}{5(d_1^2+d_2^2+d_3^2)} \frac{\Delta p}{\mu L}$. We test our approximation on a family of isosceles triangles with base $d$ and height $h$ by comparing with numerical solutions obtained from COMSOL Multiphysics (COMSOL, Los Angeles, USA) (Fig. \ref{Triangular_Cond}). Similar to the rectangular cross-section channel, our approximation is highly accurate when the ratio of height to base length is $\lesssim 2$. In fact, when the cross-section is an equilateral triangle ($\frac{h}{d} = \frac{\sqrt{3}}{2} \approx 0.866 $) we obtain $Q = \frac{\sqrt{3}d^4 \Delta p}{320\mu L}$ which is the exact conductance \citep{triangular}. The approximation starts to diverge for slot like channels, likely again because the approximate velocity profile has a forced dependence upon distance from the base that matches less and less well to the real velocity profile.
\begin{figure}
    \centering
    \includegraphics[width = 0.5 \textwidth]{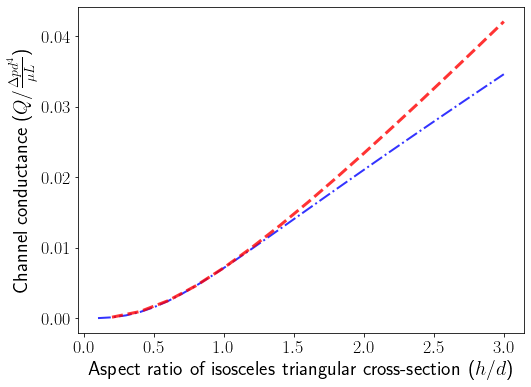}
    \caption{For isosceles triangle shaped channels, conductance approximation using excess dissipation (blue dot-dashed curve) agrees closely with a numerical solution (red dashed curve) for aspect ratio $h/d<2$.}
    \label{Triangular_Cond}
    \end{figure} 
\section{Using Lagrange multipliers to expand the space of test functions}\label{Lagrange_Section}
Hitherto, we have estimated the conductances of channels by minimizing $\mathcal{H}$ over test functions that are forced to obey all velocity boundary conditions imposed upon $\partial \Omega_U$. These test functions may not have velocity profiles similar to the real flow, and it may, in fact, be difficult to obtain non-trivial test velocities that satisfy all prescribed boundary conditions. We therefore partition $\partial \Omega_U$ into subdomains $\partial \Omega_V$ and $\partial \Omega_L$. On $\partial \Omega_V$ we impose that test functions must obey any prescribed velocity boundary conditions. Meanwhile, on $\partial \Omega_L$ we define a Lagrange multiplier $\boldsymbol{\lambda} : \partial \Omega_L \to \mathbb{R}^3$ and add to $\mathcal{H}$ a constraint term enforcing velocity boundary conditions on $\partial \Omega_L$:
\begin{equation}\label{Dissipation_Lagrange}
    \mathcal{H}_L[\mathbf{u},\boldsymbol{\lambda}] = 2\mu \int_{\Omega}\mathbf{E}:\mathbf{E} \: dV - 2\int_{\partial \Omega_T}\mathbf{T}\cdot \mathbf{u} \: dS + \int_{\partial \Omega_L} \boldsymbol{\lambda}\cdot(\mathbf{u}-\mathbf{U}) \: dS
\end{equation}
Now, if we compare two sets of functions $(\mathbf{u}, \mathbf{\boldsymbol{\lambda}})$ and $(\mathbf{u}^{\star}, \mathbf{\boldsymbol{\lambda}}^{\star})$ where $\mathbf{u}$ is the solution of Stokes' equations, we obtain, on appealing to the weak form (\ref{Weak_form_stokes})
\begin{align}\label{Dissipation_Lagrange_Deriv}
    \mathcal{H}_L[\mathbf{u}^{\star},\boldsymbol{\lambda}^{\star}]-\mathcal{H}_L[\mathbf{u},\boldsymbol{\lambda}] &= 4\mu \int_{\Omega} \mathbf{E}^{\star}:(\mathbf{E}-\mathbf{E}) \: dV + 2\mu \int_{\Omega} (\mathbf{E}-\mathbf{E}^{\star})(\mathbf{E}-\mathbf{E}^{\star}) \: dV \\
    &\quad -2\int_{\partial \Omega_T}\mathbf{T}\cdot(\mathbf{u}^{\star}-\mathbf{u}) + \int_{\partial \Omega_L}\left[ \boldsymbol{\lambda}^{\star}\cdot(\mathbf{u}^{\star}-\mathbf{U})-\boldsymbol{\lambda}\cdot(\mathbf{u}-\mathbf{U})\right]\: dS \notag \\
    &= 2\mu \int_{\Omega} (\mathbf{E}-\mathbf{E}^{\star})(\mathbf{E}-\mathbf{E}^{\star}) \: dV \notag + 2 \int_{\partial \Omega_L} \mathbf{n}\cdot \boldsymbol{\sigma}\cdot (\mathbf{u}^{\star}-\mathbf{u}) \: dS \notag\\
    &\quad + \int_{\partial \Omega_L}\left[ \boldsymbol{\lambda}^{\star}\cdot(\mathbf{u}^{\star}-\mathbf{U})-\boldsymbol{\lambda}\cdot(\mathbf{u}-\mathbf{U})\right]\: dS \notag 
\end{align}
Hence if $\mathbf{u} =\mathbf{U}$ on $\partial \Omega_L$ and $\boldsymbol{\lambda} = -2\mathbf{n}\cdot \boldsymbol{\sigma}$ then 
\begin{equation}
    \mathcal{H}_L[\mathbf{u}^{\star},\boldsymbol{\lambda}^{\star}]-\mathcal{H}_L[\mathbf{u},\boldsymbol{\lambda}] = 2\mu\int_{\Omega} (\mathbf{E}-\mathbf{E}^{\star})(\mathbf{E}-\mathbf{E}^{\star}) \: dV + \int_{\partial \Omega_L} (\boldsymbol{\lambda}^{\star}-\boldsymbol{\lambda})(\mathbf{u}^{\star}-\mathbf{U}) \: dS 
\end{equation}
So $\mathbf{u}^{\star} = \mathbf{u}$, $\boldsymbol{\lambda}^{\star} = \boldsymbol{\lambda}$ is a saddle point of the function $\mathcal{H}_{L}$. We then identify critical points of $\mathcal{H}_L$ by using test functions for $\boldsymbol{\lambda}$ and $\mathbf{u}$; note that we can no longer assure that our approximate velocity fields will produce lower bounds on the conductance. To test how close the approximations may be to the real channel conductance in a real example, we revisit the approximation for flow in a rectangular channel. 

\subsection{Rectangular cross-section channel revisited}

As discussed at the end of Section \ref{Rectangular}, while our conductance approximation for the rectangular cross-section channel is accurate when the aspect ratio is near $1$, the approximation does not provide an asymptotically correct behavior for the slot-like channel when the aspect ratio is high. This results from our assumption that the test functions needed to be quadratic in $x$, whereas in reality, we expect the velocity to be constant in $x$ up to an $O(h)$-sized region near the channel side walls. To approximate the velocity profile at high aspect ratios, we consider fluid profiles $u(x,y) = a-cx^2-dy^2$, and a pair of constant Lagrange multipliers, $\boldsymbol{\lambda} = \lambda_1\boldsymbol{e}_z$ on $x=\pm \frac{l}{2}$ and $\lambda=\lambda_2\boldsymbol{e}_z$ on $y=\pm\frac{h}{2}$.

Differentiation of Eqn. (\ref{Dissipation_Lagrange}) with respect to $\lambda_1$ sets the flux on $x = \pm \frac{l}{2}$ equal to 0. That is, 
\begin{equation}
    \int_{-\frac{h}{2}}^{\frac{h}{2}} \left(u \big|_{x = \pm \frac{l}{2}}\right)dy  = \int_{-\frac{h}{2}}^{\frac{h}{2}} \left( a - \frac{cl^2}{4}-dy^2 \right)dy =  0 
\end{equation}
which gives $a = \frac{cl^2}{4} + \frac{dh^2}{12}$. Similarly, setting the derivative with respect to $\lambda_2$ to 0, requires setting the fluxes along the boundaries $y = \pm \frac{h}{2}$ to zero, so that $a = \frac{cl^2}{12} + \frac{dh^2}{4}.$ We may solve for $c$ and $d$ in terms of $a$ as $c = \frac{3a}{l^2}$ and $d = \frac{3a}{h^2}$. Thus, the excess dissipation becomes:
\begin{align}
    \mathcal{H}[\mathbf{u}] &= \frac{\mu L}{3} \left(c^2l^3h + d^2h^3l\right) - \Delta p L h l a  = \frac{\mu L}{3}\left(\frac{9a^2h}{l} + \frac{9a^2l}{h}\right) - \Delta p L h l a \\ 
    &= 3\mu L \left(\frac{h^2+l^2}{hl}\right)a^2 - \Delta p L h l a \notag
\end{align}
which is minimized by $a = \frac{1}{6} \frac{\Delta p}{\mu L} \frac{h^2l^2}{h^2+l^2}.$ Hence, $Q = \frac{1}{2} a l h =  \frac{1}{12}\frac{\Delta p}{\mu L} \frac{h^3l^3}{h^2 + l^2}.$ This formula is asymptotic (within an aspect ratio independent constant) to the exact conductance of the channel in the high aspect ratio limit (see Fig. \ref{Rectangular_Cond_Revisit}).
\begin{figure}
    \centering
    \includegraphics[width = 0.5 \textwidth]{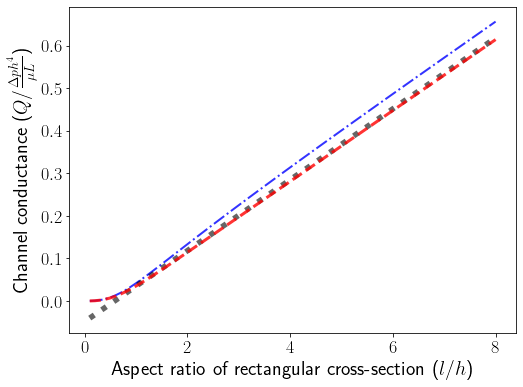}
    \caption{In a rectangular channel using Lagrange multipliers to relax the channel wall velocity constraint produces an asymptotically accurate (blue dot-dashed curve) approximation to the exact conductance (red dashed curve), in the limit of slot like channels. No slip on the widest walls, and optimal selection of the Lagrange multiplier on the narrowest walls gives an approximate conductance that is accurate for all $l/h\gtrsim 1$ (black dotted curve).  }
    \label{Rectangular_Cond_Revisit}
\end{figure} 

An even better approximation can be obtained for slot-like channels by emphasizing that the solution will be approximately $x-$invariant, except near the side walls of the channel: $u(x,y) \approx c(\frac{h^2}{4}-y^2)$. Note that this family of test functions satisfies no-slip on the boundaries $y = \pm \frac{h}{2}$ but not on $x = \pm \frac{l}{2}$. Our approach is to assign Lagrange multipliers at $x = \pm \frac{l}{2}$. However, formally optimizing $\mathcal{H}_L$ with respect to these Lagrange multipliers would produce $\int (u|_{x=\pm l/2})\,dy = 0$ leading to the trivial critical point $c=0$. But, recalling that at the critical point, the Lagrange multiplier is simply the viscous shear due to the flow, and noting that this shear will converge as $l/h$ becomes larger, we choose constant  $\boldsymbol{\lambda}/\left(\frac{\Delta ph}{L}\right) = k\mathbf{e_z},$ and leave $k$ to be determined. In so doing our Lagrange multiplier term becomes simply a penalty term within $\mathcal{H}_L$:
\begin{align}\label{Dissipation_Lagrange_Rectangular}
\mathcal{H}_L[\mathbf{u},\boldsymbol{\lambda}] &= \mathcal{H}[\mathbf{u}] + 2\int_{0}^{L}\int_{-\frac{h}{2}}^{\frac{h}{2}}\left(\frac{hk\Delta p}{L}\right)\cdot u \: dydz = \left(\frac{h^3l\mu L}{3}\right)c^2 + \left(\frac{h^4k\Delta p-h^3l\Delta p}{3}\right)c
\end{align}
Viewing $\mathcal{H}_L$ as a function of $c$ yields a critical value when $c = \left(\frac{l-hk}{2l}\right)\frac{\Delta p}{\mu L}$, and this value of $c$ gives flux $Q = \frac{h^3(l-hk)}{12}\frac{\Delta p}{\mu L}$. What is an optimal choice for $k$? We equate the new expression for the channel conductance with our previous approximation in Section \ref{Rectangular} when the aspect ratio is $1$. That is, 
\begin{align*}
\left[\frac{h^3(l-hk)}{12}\frac{\Delta p}{\mu L} = \frac{5h^3l^3}{72(h^2+l^2)}\frac{\Delta p}{\mu L} \right]\bigg|_{\frac{l}{h}=1} \Longrightarrow \quad k = \frac{7}{12} 
\end{align*}
Our new approximation using the Lagrange multiplier as a penalty term accurately approximates the channel conductance for all $h/L>1$, including in the limit of slot-like channels (Fig. \ref{Rectangular_Cond_Revisit}).

\section{Approximating flow in a cylindrical channel with slip boundary conditions}\label{Mixed_Boundary}

\subsection{Theory for a channel with slip boundary conditions}

In this section, we extend our framework for approximating microchannel flows to include more complex boundary conditions. Advances in microfluidic engineering now allow for construction of channels in which fluids slip along one or more channel walls \citep{Lauga2007}. We use our approximation method to develop approximate conductance laws for cylindrical channels with bands of slip and no-slip boundary conditions, a scenario for which exact conductance laws have previously been derived \citep{lauga2003effective}. We first modify the minimum excess dissipation principle with Lagrange multipliers, from Eqn. (\ref{Dissipation_Lagrange}), to include when the boundary of $\Omega$ can be divided into 3 sets; $\partial \Omega_T$ on which tractions $\mathbf{n} \cdot \boldsymbol{\sigma} = \mathbf{T}$ are specified, $\partial \Omega_U$ on which $\mathbf{u} = \mathbf{U}$ is specified and $\partial \Omega_S$ (the slip boundary) upon which $\mathbf{u} \cdot \mathbf{n} = 0$ and $(\mathbf{n} \cdot \boldsymbol{\sigma})\cdot( \mathbb{1}-\mathbf{n}\mathbf{n}) = \mathbf{S}$, meaning tangential components of the stress are known and there is no normal flow. As before, we release ourselves from the constraint of satisfying all the velocity boundary conditions by subdividing $\partial \Omega_U$ into $\partial \Omega_V$, where $\mathbf{u}$ is prescribed, and $\partial \Omega_L$, where we define a Lagrange multiplier to satisfy an approximate form of the velocity constraint. Then, following the same derivation steps as for equation (\ref{Dissipation_Lagrange_Deriv}), we obtain that 
\begin{equation}\label{Dissipation_Lagrange_Slip}
    \mathcal{H}_L[\mathbf{u}, \boldsymbol{\lambda}] = 2\mu \int_{\Omega}\mathbf{E}:\mathbf{E} \: dV - 2\int_{\partial \Omega_T}\mathbf{T}\cdot \mathbf{u} \: dS -2\int_{\partial \Omega_S}\mathbf{S} \cdot \mathbf{u}\: dS + \int_{\partial \Omega_L} \boldsymbol{\lambda}\cdot(\mathbf{u}-\mathbf{U}) \: dS
\end{equation}
has a saddle point when $\mathbf{u}$ is the Stokes flow and $\boldsymbol{\lambda} = -2\mathbf{n}\cdot\boldsymbol{\sigma}$, and the saddle point is a minimum point if $\partial \Omega_L = \emptyset$. 

In the special case where $\partial \Omega_S$ is a perfect slip boundary ($\mathbf{S} = \mathbf{0}$) we find that the Stokes flow minimizes $\mathcal{H}[\mathbf{u}] = 2\mu \int_{\Omega} \mathbf{E}:\mathbf{E} \: dV - 2\int_{\partial \Omega_T} \mathbf{T}\cdot\mathbf{u} \: dS$ among all velocity fields satisfying no slip boundary conditions on $\partial \Omega_U$ and no through flow boundary conditions on $\partial \Omega_S$. We also note that minimization of excess dissipation gives us a comparison principle for channels with different areas of slip boundary. Consider, for example, two channels $\Omega_1, \Omega_2$ with the same cross-section shape and applied pressure drop $\Delta p$. Suppose the slip regions $\partial \Omega_{S_1}$ of $\partial \Omega_{1}$ are completely contained within the slip regions $\partial \Omega_{S_2}$ of $\partial \Omega_{2}$. Suppose $\mathbf{u}_1, \mathbf{u}_2$ are the corresponding Stokes flow and $\mathcal{H}_{1},\mathcal{H}_2$ are the associated excess dissipation functionals for each channel. Then by the minimization principle, since $\mathbf{u}_1$ satisfies the no through flow on $\partial \Omega_{S_2} \subset \partial \Omega_{S_1} \cup \partial \Omega_{U_1}$ and no slip on $\partial \Omega_{U_2} \subset \partial \Omega_{U_1}$, it follows that 
\begin{equation}\label{slip_dissipation_ie}
    \mathcal{H}_2[\mathbf{u}_2] \leq \mathcal{H}_2[\mathbf{u}_1] = \mathcal{H}_1[\mathbf{u}_1]
\end{equation}
where we have made use of the fact that $\mathbf{T}_1 = \mathbf{T}_2$ and $\partial \Omega_{T_1} = \partial \Omega_{T_2}$ to identify $\mathcal{H}_1$ and $\mathcal{H}_2$. In the case of slip boundaries, external work is done only by the surface tractions applied on $\partial \Omega_{T}$, and 
\begin{equation}
    \text{rate of dissipation} = 2\mu \int_{\Omega} \mathbf{E}:\mathbf{E} \: dV = \int_{\partial \Omega_T} \mathbf{T} \cdot \mathbf{u} \: dS
\end{equation}
So our inequality (\ref{slip_dissipation_ie}) may be rewritten as 
\begin{align}
    &-\int_{\partial \Omega_{T}}\mathbf{T}\cdot\mathbf{u}_2 \: dS \leq - \int_{\partial \Omega_{T}}\mathbf{T}\cdot \mathbf{u}_1 \notag \\
    &\text{or} \quad \quad \quad  Q_2 \Delta p \geq Q_1 \Delta p
\end{align}
In other words, the conductance of a channel is always expanded by expanding the zones of slip along the channel walls.

In the following sections, we consider an infinitely long circular-cylindrical channel patterned by perfect slip and no slip surfaces on two elementary configurations we call longitudinal and transverse (see Fig. \ref{config}). In the longitudinal configuration, the slip region runs parallel to the flow direction, taking the form of a stripe occupying an arc of angle $0 \leq \psi < 2\pi$ along the boundary of the channel. In the second configuration, the slip regions are distributed transverse to the flow direction, taking the form of rings with width $l$ periodically distributed along the channel with period $L$. 

\begin{figure}
    \centering
    \includegraphics[width = 0.5 \textwidth]{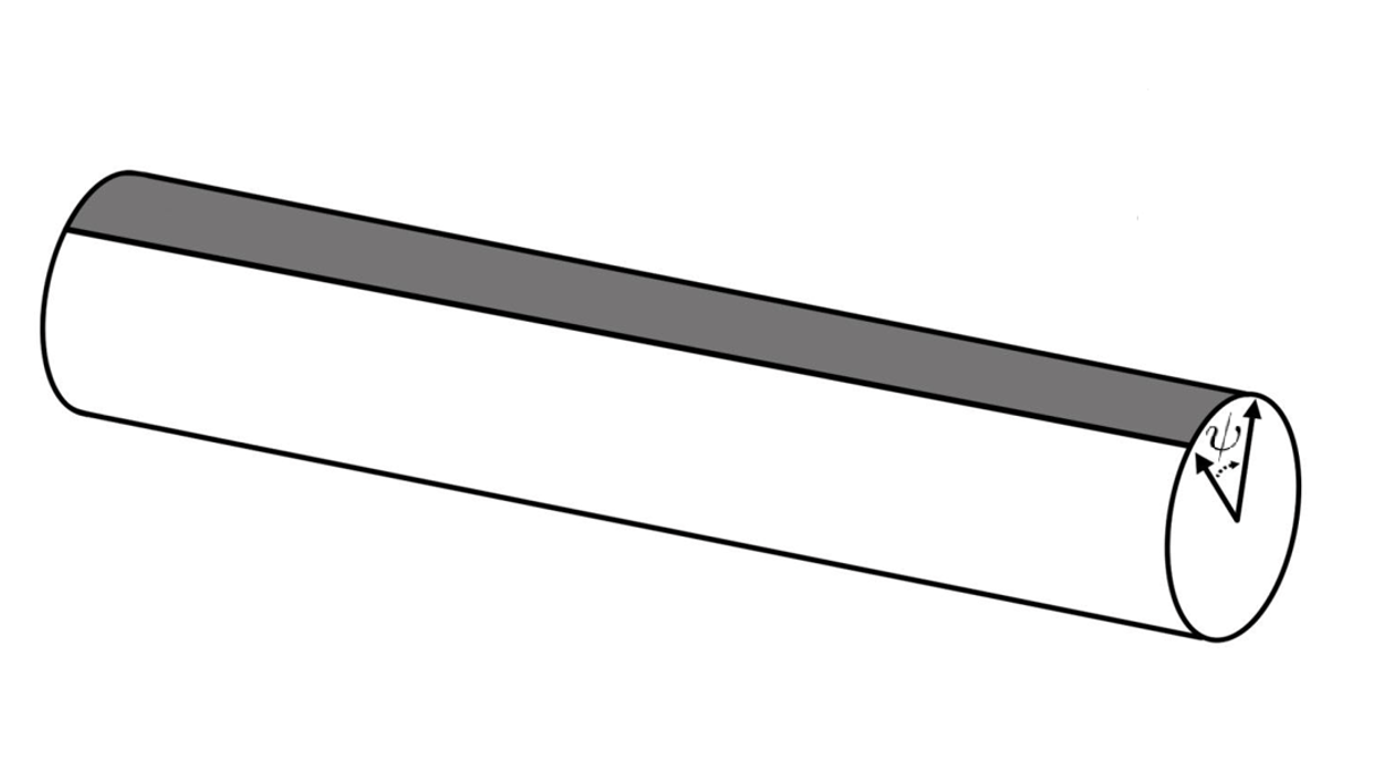}(a)
    \includegraphics[width = 0.5 \textwidth]{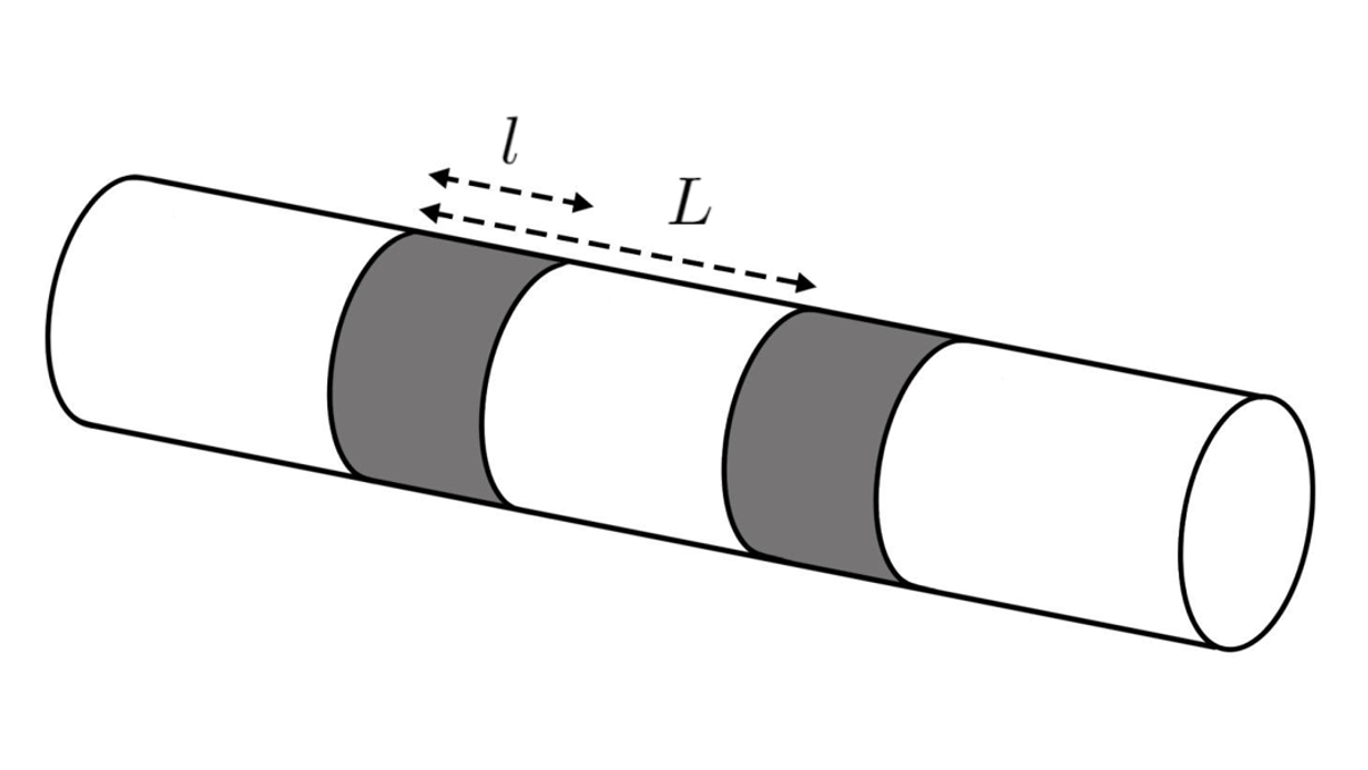}(b)
    \caption{Two elementary configurations of an infinitely long circular-cylindrical channel patterned by perfect slip (shaded region) and no slip surfaces: (a) longitudinal configuration (b) transversal configuration.} 
    \label{config}
\end{figure} 

As usual, for both configurations, we prescribe periodic velocity boundary conditions and stress boundary conditions $\mathbf{n} \cdot \boldsymbol{\sigma} = -\Delta p \mathbf{n}$ at $z = 0$ and $\mathbf{n} \cdot \boldsymbol{\sigma} = 0$ at $z = L$. We then consider the Lagrange multiplier acting entirely along the no slip region $\partial \Omega_{U} = \partial \Omega_{L} = \partial \Omega_{\text{no slip}}$. 

\subsection{Longitudinal configuration}\label{Longitudinal}
In the longitudinal configuration (Fig. \ref{config}(a)), the flow is translationally invariant. We shall work in cylindrical coordinates $(r,\theta,z)$ and assume that the no-slip region occupies the region $ |\theta| < \pi - \frac{\psi}{2}$ while the slip region occupies $|\theta| > \pi - \frac{\psi}{2}$. If the channel radius is $a$, we may non-dimensionalize velocities by $\frac{\Delta p}{\mu L}a^2$, pressures by $\Delta p$, and lengths by $a$, so that Eqn. (\ref{Prismatic_Deriv}) becomes $\nabla^2 u = -1$, and the underlying boundary conditions become:
\begin{align}\label{longitudinal_bcs}
    u &= 0 \quad \text{for } r = 1, \quad 0 \leq z \leq L, \quad |\theta| < \pi - \frac{\psi}{2} \\
    \frac{\partial u}{\partial r} &= 0 \quad \text{for  } r = 1, \quad 0 \leq z \leq L, \quad \pi - \frac{\psi}{2} \leq |\theta| \leq \pi \notag
\end{align}
When $\psi=0$, the boundary conditions are entirely no slip. Therefore, we expect that the dimensional velocity profile and channel conductance would obey the Hagen-Poiseuille law, $u_{HP} = \frac{1}{4}\left(1-r^2\right)$, $Q_{HP} = \frac{\pi}{8}$. Conversely, in the case where $\psi \to 2\pi$, so that the entire channel is slip, the velocity profile will approach plug flow, and the conductance will diverge. 

\subsubsection{Approximation using combination of Hagen-Poiseuille flow with harmonic basis}

Since the harmonic functions $r^n \cos{n \theta}$ and $r^n \sin{n \theta}$ form a complete basis, we generate a family of test functions by adding to the Hagen-Poiseuille flow the first two harmonic functions $\mathbf{u} = u(r,\theta) \mathbf{e_z} = (c_1 + c_2r^2 + c_3r \cos{\theta} + c_4r^2 \cos{2 \theta})\mathbf{e_z},$ where $c_1,c_2,c_3,c_4$ are constant parameters to be determined and our choice of cosine is dictated by the symmetry of the problem. In addition, we consider two schemes of Lagrange multiplier $\boldsymbol{\lambda}_1 = k\mathbf{e_z}$ and $\boldsymbol{\lambda}_2 = (k+ j\cos{\frac{\theta}{2}})\mathbf{e_z}$ defined on the no-slip region $|\theta| < \pi - \frac{\psi}{2}$ where $k$ and $j$ are constant parameters to be determined; an additional term $j \cos{\frac{\theta}{2}}$ is included to mimic the fact that at the extremum, $\boldsymbol{\lambda}$ is proportional to the shear stress: $\boldsymbol{\lambda}$ takes its maximum value at the center of no-slip region $\theta = 0$, and we expect it to decrease away from the center of the no-slip band. We may calculate the (non-dimensionalized) modified excess dissipation:
\begin{small}
\begin{align}
    \tilde{\mathcal{H}}_{L}[\mathbf{u},\boldsymbol{\lambda}_1] &= (2c_2^2+c_3^2+2c_4^2)\pi - (2c_1+c_2)\pi + k\left((c_1+c_2)(2\pi - \psi) + 2c_3 \sin{\frac{\psi}{2}}-c_4\sin{\psi}\right) \\
    \tilde{\mathcal{H}}_{L}[\mathbf{u},\boldsymbol{\lambda}_2] &= (2c_2^2+c_3^2+2c_4^2)\pi - (2c_1+c_2)\pi + (2\pi-\psi)(c_1+c_2)k + (4c_1+4c_2+2c_3)j\cos{\frac{\psi}{4}} \notag \\
    &\quad -\frac{2}{3}(c_3+c_4)j\cos{\frac{3\psi}{4}}+\frac{2}{5}c_4j\cos{\frac{5\psi}{4}} + 2c_3k\sin{\frac{\psi}{2}}-c_4k\sin{\psi}
\end{align}
\end{small}
The values of $c_1,c_2,c_3,c_4,k$, and $j$ are then obtained at the critical point of $\tilde{\mathcal{H}}_{L}[\mathbf{u},\boldsymbol{\lambda}]$ by solving the system of equations $\frac{\partial \tilde{\mathcal{H}}_{L}}{\partial c_1} = \frac{\partial \tilde{\mathcal{H}}_{L}}{\partial c_2} = \frac{\partial \tilde{\mathcal{H}}_{L}}{\partial c_3} = \frac{\partial \tilde{\mathcal{H}}_{L}}{\partial c_4} = \frac{\partial \tilde{\mathcal{H}}_{L}}{\partial k} = \frac{\partial \tilde{\mathcal{H}}_{L}}{\partial j} = 0$ using the computer algebra software Mathematica (Wolfram Research, Champaign, USA). We then calculate the channel conductance $Q$ from our approximated flow
\begin{small}
    \begin{align}
    Q_{\text{longitudinal}}[\mathbf{u},\boldsymbol{\lambda}_1] &= \frac{\pi}{8}+\frac{16\pi\sin^2{\frac{\psi}{2}}+2\pi\sin^2{\psi}}{4(2\pi-\psi)^2} \\
    Q_{\text{longitudinal}}[\mathbf{u},\boldsymbol{\lambda}_2] &= -\bigg( \bigg(  \pi \bigg( -2109 - 555 (2\pi - \psi)^2 + 96(-1 + 6(2\pi-\psi)^2)\left( \cos{\frac{\psi}{2}}\right) \\ \notag
    &- 6(-305 + 27(2\pi-\psi)^2 )\cos{\psi} + (80 + 48(2\pi- \psi)^2)\left(\cos{\frac{3\psi}{2}}\right) \\ \notag
    &+ 277 \cos{2 \psi} - 9(2\pi - \psi)^2 \cos{2 \psi} + 16\cos{ \frac{5 \psi}{2}} + 2 \cos{3 \psi}     \\ \notag
    &+ 2640(2\pi-\psi)\left(\sin{\frac{\psi}{2}}\right)- 840(2\pi - \psi)\sin{\psi} + 240(2\pi - \psi) \left(\sin{\frac{3 \psi}{2}}\right) \\ \notag
    &- 90(2\pi - \psi)\sin{2 \psi} \bigg) \bigg) \bigg/ \\ \notag
    &\bigg( 24\bigg(675 + 185(2\pi - \psi)^2 - 192(2\pi - \psi)^2\left(\cos{\frac{\psi}{2}}\right) \\ \notag
    &+ 6(-100 + 9(2\pi - \psi)^2)\cos{\psi} - 16(2\pi-\psi)^2\left(\cos{\frac{3\psi}{2}}\right) \\ \notag
    &- 75\cos{2\psi} + 3(2\pi - \psi)^2 \cos{2 \psi} - 880 (2\pi - \psi) \left(\sin{\frac{\psi}{2}}\right) \\ \notag
    &+ 280(2\pi - \psi) \sin{\psi} - 80(2\pi -\psi) \left(\sin{\frac{3\psi}{2}}\right) + 30( 2\pi-\psi) \sin{2\psi} \bigg) \bigg) \bigg)
    \end{align}
\end{small}
and compare the result with the actual fluid conductance obtained from solving Eqn. (\ref{Prismatic_Deriv}) over different values of $\psi$, (Fig. \ref{longitudinal_fig}): 
\begin{figure}
\centering
\includegraphics[width = 0.5 \textwidth]{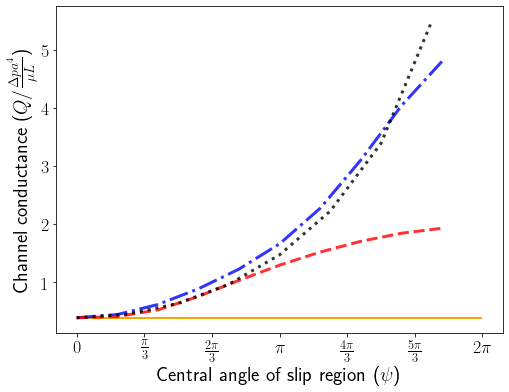}
\caption{Fluid conductance approximation for the longitudinal slip configuration using modified excess dissipation and two schemes of Lagrange multiplier terms $\boldsymbol{\lambda}_1$ (blue dot-dashed curve) and $\boldsymbol{\lambda}_2$ (red dashed curve). The $\lambda_2$ approximation asymptotically captures the numerically calculated increase of conductance (dotted black curve) away from the Hagen-Poiseuille estimate (orange line) as slip regions are first introduced, though, likely serendipitously, the simpler approximation renders the increasing conductance more accurately at moderate values of $\psi$.}
\label{longitudinal_fig}
\end{figure} 
The approximate conductance obtained with Lagrange multiplier $\boldsymbol{\lambda}_2$  is asymptotically correct when $\psi \to 0$, so the channel is mostly no slip, however, perhaps serendipitously, the fluid conductance obtained using $\boldsymbol{\lambda}_1$, though not asymptotically exact for small $\psi$ matches the numerically obtained conductance over a larger range in $\psi$ (Fig. \ref{longitudinal_fig}). However, neither approximation agrees with the true conductance when it diverges as $\psi\to 2\pi$.

To improve upon our approximation we try to approximate the flow field and shear stress (since at the extremum of $\mathcal{H}_L$, $\boldsymbol{\lambda}$ is proportional to the shear stress) more precisely, in the limit of large slip bands. We find empirically, through COMSOL simulation, that the shear stress is close to uniform across the no-slip panel but diverges at the edge of the slip band (Fig. \ref{Shear_Stress}(a)). To model this divergence we perform a local analysis of the flow equations at the junction of slip and no-slip panels. Near to this junction we treat the channel as locally flat, defining a distance, $\rho$, from the junction line and an angle $\varphi$ made with the channel wall (Fig. \ref{Shear_Stress}(b)). 

\begin{figure}
    \centering
    \includegraphics[width = 0.5 \textwidth]{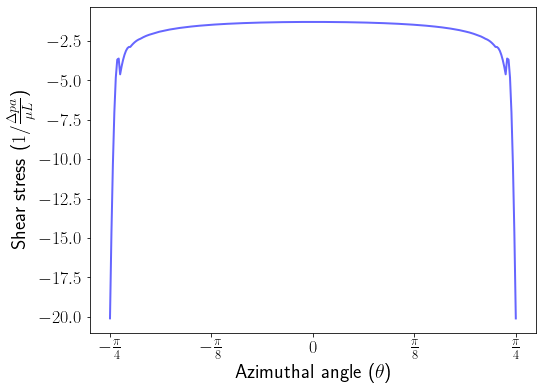}(a)
    \includegraphics[width = 0.4 \textwidth]{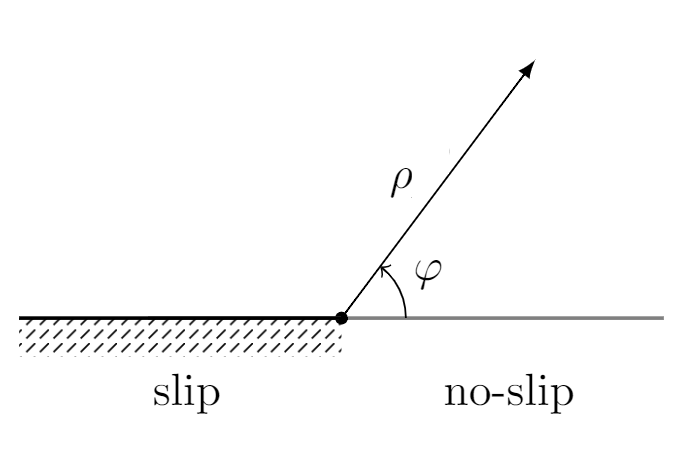}(b)
    \caption{More careful modeling of the shear stresses along the no-slip panel allows us to approximate the flow field in the limit of small slip bands. (a) Numerically computed shear stress along the boundary of a longitudinal configuration channel with slip region of central angle $\psi = \frac{3\pi}{2}$ so that the no-slip region occupies the boundary $-\frac{\pi}{4} < \theta  < \frac{\pi}{4}$. (b) Coordinate system used for local analysis of the shear stress at the interface between the slip and no-slip boundary.} 
    \label{Shear_Stress}
\end{figure} 

Now, in the dimensionless form of the equations we have $\nabla^2 u = -1$ and so $u = -\frac{1}{4}\rho^2 + u'$ where $u'$ is a harmonic function. We know that harmonic functions are a sum of terms $\rho^{\alpha} \sin{\alpha \varphi}, \rho^{\alpha} \cos{\alpha \varphi}$. On each order of $\rho$, we impose boundary conditions:
\begin{align}
    u' &= 0 \quad \text{at} \quad \varphi = 0 \\
    \frac{1}{\rho}\frac{\partial u'}{\partial \varphi} &= 0 \quad \text{at} \quad \varphi = \pi
\end{align}
Taken together, our boundary conditions admit only harmonic functions of the form $\rho^{\alpha} \sin{\alpha \varphi}$ with $\alpha \pi \in (\mathbb{Z}+\frac{1}{2})\pi$. Restricting to velocity fields that remain bounded as $\rho \to 0$, we obtain that 
\begin{equation}
    u' = c_1\rho^{\frac{1}{2}}\sin{\frac{1}{2}\varphi}+ c_2 \rho^{\frac{3}{2}}\sin{\frac{3}{2}\varphi} - \frac{1}{4}\rho^2 + o(\rho^2) 
\end{equation}
Hence, the shear stress on the no-slip boundary at $\varphi = 0$ is equal to:
\begin{equation}
    \frac{1}{\rho}\frac{\partial u'}{\partial \varphi}|_{\varphi = 0} = \frac{1}{2}c_1\rho^{-\frac{1}{2}}+\frac{3}{2}c_2\rho^{\frac{1}{2}}+o(\rho)
\end{equation}
In particular, the shear stress diverges with the inverse square root of distance from the boundary between slip and no-slip panels. The results of our local analysis can be recapitulated by expanding the exact solution of the velocity field in \citet{exact_longitudinal} around the panel edge. We note that although the shear stress diverges at the edge of the no-slip panel, the total shear force obtained by integrating the shear stress remains finite.

Since the boundary between no-slip region and slip-region occurs at $\theta = \pm(\pi - \frac{\psi}{2})$, we posit two new types of Lagrange multipliers of the form  $\boldsymbol{\lambda}_3 = \frac{k}{\sqrt{(\pi-\frac{\psi}{2})^2-\theta^2}}\mathbf{e_z}$ and  $\boldsymbol{\lambda}_4 = k\left(\frac{1}{\sqrt{(\pi-\frac{\psi}{2})^2-\theta^2}} + \sqrt{(\pi-\frac{\psi}{2})^2-\theta^2}\right)\mathbf{e_z}$. Although more complicated Lagrange multipliers can be introduced, such as $\lambda_3 + j$ $\frac{k}{\sqrt{(\pi-\frac{\psi}{2})^2-\theta^2}} + j\sqrt{(\pi-\frac{\psi}{2})^2-\theta^2}$, adding additional unknown parameters produces over-determined systems that must be solved to obtain the critical point. Thus, we proceed with these two types of Lagrange multiplier $\boldsymbol{\lambda}_3,\boldsymbol{\lambda}_4$ and obtain the relation between the flow and channel conductance:   
\begin{small}
\begin{align}
    Q_{\text{longitudinal}}[\mathbf{u},\boldsymbol{\lambda}_3] &= \frac{\pi}{8}+ \pi J_{0}\left[\frac{\psi}{2}-\pi\right]^2+ \frac{\pi}{2}J_{0}\left[\psi-2\pi\right]^2  \\
    Q_{\text{longitudinal}}[\mathbf{u},\boldsymbol{\lambda}_4] &= \frac{\pi}{8}+ \frac{2\pi}{(8+(2\pi-\psi)^2)^2}\bigg((2\pi-\psi)J_{1}[\psi-2\pi]-4J_{0}[\psi-2\pi]\bigg)^2 \notag \\
    &\quad + \frac{16\pi}{(8+(2\pi-\psi)^2)^2}\left((2\pi-\psi)J_{1}\left[\frac{\psi}{2}-\pi\right]-2J_{0}\left[\frac{\psi}{2}-\pi\right]\right)^2 
\end{align}
\end{small}
where $J_{n}[x]$ are Bessel functions of the first kind. As seen in Fig. \ref{inv_sqrt_Lagrange}, the two new schemes yield good approximations to the numerically computed conductance law across the full range of assayed values of $\psi$, albeit with an error that continues to increase as we approach the singular limit $\psi\to 2 \pi$.
\begin{figure}
\centering
\includegraphics[width = 0.50 \textwidth]{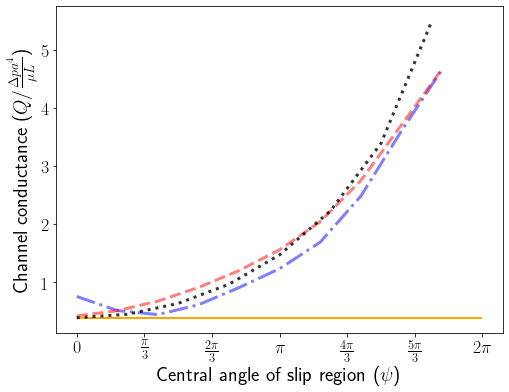}
\caption{Approximating the conductance of a channel with longitudinal using Lagrange multipliers $\lambda_3$ (one term, blue dot-dashed curve) and $\lambda_4$ (two terms, red dashed curve) models for the singular stress in the slip region, produces closer approximations to the numerically obtained channel conductance (black dotted curve). Again, the Hagen-Poiseuille conductance is included as a reference (orange curve).}  
\label{inv_sqrt_Lagrange}
\end{figure}

\subsection{Transverse configuration}\label{Transversal}
Next, we consider a periodic array of stream-perpendicular panels of slip, specifically, we consider a cylinder of length $L$, extended periodically, in which the slip regions arranged transverse to the flow direction, having width of $l$. To ensure consistency with our setups for the longitudinal configuration, we assume that the no-slip boundary condition occupies the region $|z| < \frac{L-l}{2}$ and slip boundary condition occupy the region $ \frac{L-l}{2} \leq |z| < \frac{L}{2}.$ Our boundary conditions are no longer translationally invariant, and the Stokes equations governing fluid flows may no longer be reduced to solving Poisson's equation. Instead, since the flow is axisymmetric, we may define a streamfunction $\psi(r,z)$ with 
\begin{equation}\label{streamfunction}
    u_r = \frac{1}{r}\frac{\partial \psi}{\partial z}, \quad u_z = -\frac{1}{r}\frac{\partial \psi}{\partial r}
\end{equation}
As $\psi$ is defined only up to a constant, we may impose that the streamfunction vanishes along the channel boundary. The boundary conditions for the velocity field and stress may then be written as:
\begin{align}
    \psi &= 0 \quad \text{for} \quad r=1~, \notag \\
    \frac{\partial \psi}{\partial r} &= 0 \quad \text{for} \quad r=1, \quad |z| < \frac{L-l}{2}~, \\
    \frac{\partial}{\partial r}\left(\frac{1}{r}\frac{\partial \psi}{\partial r}\right) &= 0 \quad \text{for} \quad r=1, \quad \frac{L-l}{2} \leq |z| < \frac{L}{2}~. \notag
\end{align}
Note that the streamfunction for the Poiseuille flow (i.e. if $l = 0$) is $\psi_{HP} = \frac{1}{16}(1-r^2)^2$.

Before we minimize the excess dissipation in this geometry we consider a long-wavelength approximation for the flow: if the panels of slip and no-slip boundaries were all long, relative to the diameter of the channel, then we expect that over much of the no-slip panel flow will be Poiseuille, whereas over much of the slip-panel the flow profile will be plug flow. In this sense, we may approximate the conductance by roughly assuming that the pressure only acts on the no-slip region, neglecting the transition region between the two type of flows. Hence:
\begin{equation}\label{rough_approx}
    Q_{\text{long}} \approx \frac{\pi}{8}\frac{\Delta p}{\mu (L-l)}
\end{equation}
We present the long wavelength approximation along with the results of our excess dissipation optimization in Fig. \ref{Transverse_Usual} 
\subsubsection{Approximation using combination of Hagen-Poiseuille streamfunction and harmonic basis}\label{Transversal_usual}
We incorporate the background Poiseuille flow with functions similar to the harmonic basis and posit test streamfunctions of the form 
\begin{equation}
    \psi(r,z) = (1-r^2)\left(c_1+c_2r^2+c_3r^2\cos{\frac{2\pi z}{L}}+c_4r^2\cos{\frac{4\pi z}{L}}\right)
\end{equation}
with the parameters $c_1,c_2,c_3$, and $c_4$ to be determined. Notice that our expansion includes only even powers of $r$ to ensure that $\psi$ is an even function of $r$. With the success of constant Lagrange Multiplier term in Section. \ref{Longitudinal}, we fix $\lambda = k\mathbf{e_z}$ and compute the modified energy dissipation (\ref{Dissipation_Lagrange}) in polar coordinates as:
\begin{small}
\begin{align}\label{modified_dissipation_transverse}
    \tilde{\mathcal{H}}_{L}[\mathbf{u},\boldsymbol{\lambda}]   &= 2(c_1+c_2)k(L-l)-4c_1\pi+32c_2^2L\pi + \left(\frac{16L^4+52.64L^2+64.94}{L^3}\right)c_3^2\pi \notag \\
    &\quad +\left(\frac{16L^4+210.55L^2+1039.03}{L^3}\right)c_4^2\pi + \frac{2kL}{\pi}\left(c_3-c_4\cos{\frac{l\pi}{L}}\right)\left(\sin{\frac{l\pi}{L}}\right)
\end{align}
\end{small}
Here, we may directly compute $\tilde{\mathcal{H}}_{L}[\mathbf{u},\boldsymbol{\lambda}]$ using $\psi$ by calculating $u_r$ and $u_z$ as in Eqn. (\ref{streamfunction}). We solve for each constant $c_1,c_2,c_3,c_4,k$ at the critical point of $\tilde{\mathcal{H}}_{L}[\mathbf{u},\boldsymbol{\lambda}]$ and obtain the governing non-dimensionalized conductance law:
\begin{footnotesize}
\begin{align}
    Q_{\text{transverse}} &= \frac{\pi}{8}\cdot\Bigg((263.57 + 267.05L^2 + 112.29L^4 +16.45L^6 + L^8)(L-l)^2 \notag\\
    &\quad + (13.16L^6+2.67L^8+0.20L^{10})\left(\sin^2{\frac{l\pi}{L}}\right) + (0.82L^6+0.67L^8+0.20L^{10})\left(\sin^2{\frac{l\pi}{L}}\cos^2{\frac{l\pi}{L}}\right)\Bigg)\notag\\
    &\quad \cdot \Bigg( (L-l)^2(L^4+3.29L^2+4.06)(L^4+13.16L^2+64.94)\Bigg)^{-1} \label{eq:transverse}
\end{align}
\end{footnotesize}
We consider two prototypical channels with lengths $L/a = 1$ and $L/a = 10$, vary the length of slip-region $l$, and compare with numerically computed flux obtained using COMSOL Multiphysics (Fig. \ref{Transverse_Usual}). Here, the approximation using modified excess dissipation shows a similar behavior with those we have obtained for the longitudinal configuration: the approximation does reasonably well when the boundary conditions are mostly no-slip and systematically outperforms the long wavelength approximation of Eqn. (\ref{rough_approx}). However, minimization over the chosen family of test functions under-predicts the diverging conductance obtained as $l \to L$.
\begin{figure}
\centering
\includegraphics[width = 0.9 \textwidth]{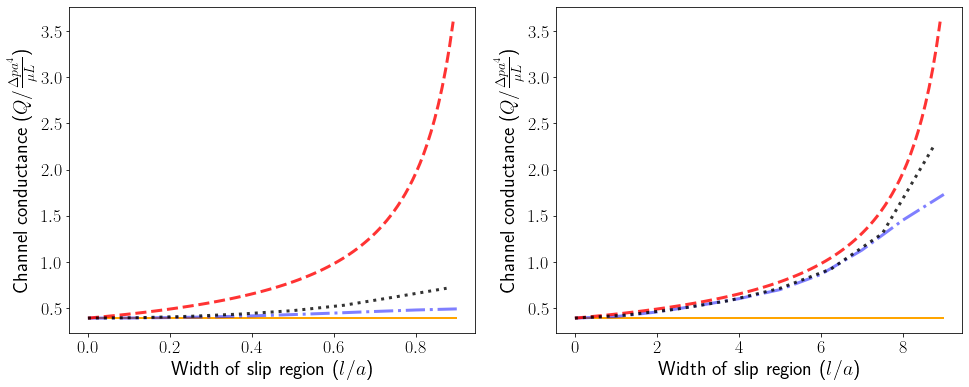}
\caption{Approximation of fluid conductances in channels with transverse patches of slip, of different widths. The panels represent different periodicity length scales $\frac{L}{a} = 1$ (left panel) and $\frac{L}{a} = 10$ (right panel). Approximating the fluid velocity using a polynomial basis produces conductances (Eqn. \ref{eq:transverse}, blue dot-dashed curve), that more accurately model the numerically computed conductance (black dotted curve), than the long wavelength approximation (Eqn. \ref{rough_approx}, red dashed curve), for moderate values of $l/L$, but not as $l/L\to 1$. In both plots, the solid orange curve is the Hagen-Poiseuille conductance, included for reference.}
\label{Transverse_Usual}
\end{figure} 

\subsubsection{Approximation using an interpolation between Poiseuille flow and plug flow}\label{Poiseuille+Plug}
The long wavelength approximation allows the conductance to diverge as $l\to L$ but over-estimates the conductance of the channel because it neglects the viscous dissipation occurring in the region where the flow transition from plug to Poiseuille flow. To help separate the two type of flows apart, we consider a different setting of the problem by assuming that the no-slip region occupies the region $0 < z < L-l$ of pipe while slip region occupies the region $L-l \leq z < L$ of the pipe instead; in this sense, the transition region between the two flows occurring at $z = L-l$. Within its transition region, we expect the stream function to interpolate from Poiseuille; $\psi \approx c(r^4-2r^2+1)$ to plug flow; $\psi \approx c(1-r^2),$ or conversely. Having the same constant $c$ in the two expressions ensures that their overall fluxes are matched. Assume that the transitions occur over a length scale $\frac{1}{k}$. Then the most divergent term in the dissipation integral is $\left(\frac{\partial u_z}{\partial z}\right)^2 \approx \mathcal{O}(k^2[u_z]^2)$, where $[u_z]$ represents the size of discontinuity in $u_z$ from plug to Poiseuille flow. Thus, the total contribution to the dissipation on integrating in $z$ over an $\mathcal{O}(\frac{1}{k})$-sized transition region is $\mathcal{O}(k[u_z]^2)$. We modify our excess dissipation to include a penalty term to represent, heuristically, the dissipation within the transition region as
\begin{align}\label{penalty_dissipation}
    \overline{\mathcal{H}}[\mathbf{u}] &= \mathcal{H}[\mathbf{u}] + k\int_{0}^{1} r\left(4c(1-r^2)-2c\right)^2\:dr \notag \\
    &= 32(L-l)\pi c^2 -4c\pi + \frac{2k}{3}c^2
\end{align}

Where minimization is performed over the parameter $c$, and there is no Lagrange multiplier since our no-slip boundary conditions are satisfied automatically.

Notice that this method would fail when $k$ is large, as the dissipation from the transition region between the two flows would overshadow the whole excess dissipation. Thus, the selection of the value $k$ is essential to the approximation. In our case, we pick value the of $k$ so that the dissipation along the interface between the Poiseuille flow and plug flow is less than the energy dissipation along the channel by an order of magnitude when $l = 0$; we choose $k = 4.8\pi L$, and find that our results are not strongly affected by choosing $k$ within 50\% of this value. The resulting conductance approximation becomes 
\begin{equation} \label{eq:transition}
    Q_{\text{(HP + Plug)}}= \frac{6\pi L}{52.8L-48l} 
\end{equation}
Introducing the penalty terms to the excess dissipation successfully reduces our estimated conductance to match more closely to the numerically obtained values (Fig. \ref{penalty_result}), for a channel with denser slip patches (smaller values of $L/a$), as the fraction of wall that is slip is increased (i.e. for $l\approx L$).
\begin{figure}
    \centering
    \includegraphics[width = 0.9 \textwidth]{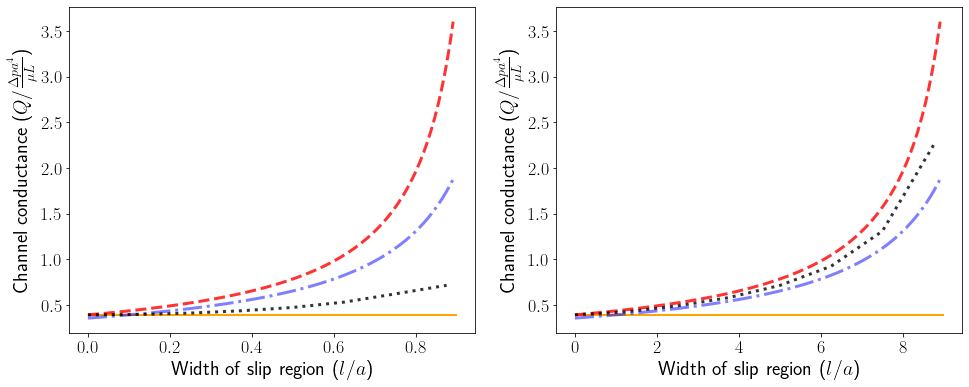} 
    \caption{Fluid conductance approximation for transverse slip patches using an additional penalty term for dissipation in the transition region between Poiseuille and plug flows approximations. The left panel is $\frac{L}{a}=1$ and the right panel is $\frac{L}{a} = 10$, and curves and colors follow Fig. \ref{Transverse_Usual}, except that the excess dissipation approximation (blue dot-dashed curve) has been replaced by Eqn. \ref{eq:transition}. The approximation is notably more accurate for more dense slip patches (comparing the left panel with Fig. \ref{Transverse_Usual})}
    \label{penalty_result}
\end{figure}  
\section{Conclusion}\label{Conclusion}

In this paper, we proposed an approximation method for the pressure-driven Stokes flow by minimizing the excess dissipation, and its variants, over a family of test functions. Starting with the excess dissipation, as formulated in \citet{keller1967extremum}, we impose a no-slip condition throughout the boundary of a channel that either has a triangular or a rectangular cross-section, and obtain a closed form approximate expression for the conductance, that empirically matches extremely closely to the numerically computed conductance when the ratio of the longest and shortest diameters of the channel is less than 2, and grows to the order of 15\% for channels with aspect ratio on the order of 4.

We have assumed in our analysis that the channel is prismatic -- that is, it has constant cross-section shape and area along its length. So long as incompressible test functions can be found, the same approximation principle could be used to calculate approximate flow fields for the many and variable channel geometries found in microfluidic devices. A good starting place for such a calculation would be to start with a long wave approximation, in which we assume a down stream velocity that at each $z$ approximates the unidirectional flow in a constant cross-section channel, and cross-channel velocities are introduced to maintain incompressibility.

In the second scheme, we consider a cylindrical pipe with two configurations, longitudinal and transverse, of banded regions of slip and no-slip, for which analytical solutions were previously derived by \citet{lauga2003effective}. In general, it is hard to come up with test functions that only vanish along one part of the boundary, while still allowing for non-zero velocities on the bands of slip; therefore, we modified our excess dissipation to include Lagrange multiplier terms that allow for an expanded set of test functions. However, the modified excess dissipation does not guarantee that the same minimization principle would still hold; rigorous arguments are needed to pinpoint the subspace of smooth vector fields which admit the Stokes flow as a minimum for the modified excess dissipation. We use this approximation to derive conductance laws that are, for both configurations of bands, in good agreement with the numerically computed conductances, so long as at least half of the channel boundary area carries a no-slip boundary condition. To resolve the failure of approximation when most boundary conditions are slip, we consider alternative test functions and penalty terms based on the transition region between the no-slip and slip boundary conditions. We note that although the analysis of circular cylindrical channels is convenient because of the existence of exact results on the conductance \citep{lauga2003effective}, our theory can be used for non-circular channels, or for channels in which the slip bands are aperiodic, or irregularly shaped.

Throughout, our approximation method has been compared with numerical computed conductances that can be obtained relatively readily using commercially available Finite Element Method Software. It might be contended that the existence of such software makes approximate solutions redundant. Nevertheless, approximate conductance laws gain utility from the fundamental limitations of analytic and numerical solutions of the flow through microchannels to expose how the flow depends quantitatively on channel shape, or the size and distribution of slip panels. By contrast, Eqn. \ref{eq:rectangular_conduct} gives a quick method for evaluating the conductance of rectangular channels of any aspect ratio, and could be used to estimate the pressure-flow relationships of microfluidic networks. Moreover, since approximate conductance laws are obtained here via a rigorous minimization principle, although there have been previous successful approximate conductance models such as \citet{bahrami2005pressure}, ours differs from them in that its accuracy can be made controllably precise.

\section*{Acknowledgement}
We thank Raymond Chu for helpful discussions, and Chuyuan Fu for correcting an error in our Eqn (6). We acknowledge financial support from the National Science Foundation, under award number DMS-2009317.
\bibliographystyle{jfm}
\bibliography{Project.bib}
\end{document}